\documentclass[11pt]{article}

\usepackage[letterpaper,margin=0.96in]{geometry}

\usepackage{bm, amsmath, amsthm, amssymb, amsfonts}

\usepackage{cite}
\usepackage{framed}
\usepackage[framemet hod=tikz]{mdframed}
\usepackage{appendix}
\usepackage{graphicx}
\usepackage{color}
\usepackage{varwidth}
\usepackage{wrapfig}
\usepackage{paralist}

\usepackage[textsize=tiny]{todonotes}

\usepackage{xspace}
\usepackage{xifthen}

\makeatletter
\renewcommand{\paragraph}{%
  \@startsection{paragraph}{4}%
  {\z@}{2.5ex \@plus 1ex \@minus .2ex}{-1em}%
  {\normalfont\normalsize\bfseries}%
}
\makeatother

\definecolor{darkgreen}{rgb}{0,0.5,0}
\definecolor{darkblue}{rgb}{0,0,0.8}
\usepackage{hyperref}
\hypersetup{
    unicode=false,          % non-Latin characters in Acrobat’s bookmarks
    colorlinks=true,        % false: boxed links; true: colored links
    linkcolor=darkblue,          % color of internal links (change box color with linkbordercolor)
    citecolor=darkgreen,        % color of links to bibliography
    filecolor=magenta,      % color of file links
    urlcolor=brown           % color of external links
}
\RequirePackage[]{silence}
\usepackage{algorithm}
\usepackage[noend]{algpseudocode}

\usepackage[nameinlink,capitalize]{cleveref}

\usepackage{thmtools}
\usepackage{thm-restate}

\newtheorem{theorem}{Theorem}[section]

\newtheorem{lemma}[theorem]{Lemma}
\newtheorem{corollary}[theorem]{Corollary}
\newtheorem{definition}{Definition}[section]

\usepackage{mathtools}

\newcommand{\calA}{\ensuremath{\mathcal{A}}}

\newcommand{\calC}{\ensuremath{\mathcal{C}}}

\newcommand{\ignore}[1]{}

\algnewcommand\algorithmicswitch{\textbf{switch}}
\algnewcommand\algorithmiccase{\textbf{case}}

\algdef{SE}[SWITCH]{Switch}{EndSwitch}[1]{\algorithmicswitch\ #1\ \algorithmicdo}{\algorithmicend\ \algorithmicswitch}%
\algdef{SE}[CASE]{Case}{EndCase}[1]{\algorithmiccase\ #1}{\algorithmicend\ \algorithmiccase}%
\algtext*{EndSwitch}%
\algtext*{EndCase}%
\newcommand{\CONGEST}{\ensuremath{\mathsf{CONGEST}}\xspace}
\newcommand{\LOCAL}{\ensuremath{\mathsf{LOCAL}}\xspace}
\newcommand{\SETLOCAL}{\ensuremath{\mathsf{SET}\text{-}\mathsf{LOCAL}}\xspace}

\newcommand{\ID}{\ensuremath{\mathrm{ID}}\xspace}

\newcommand{\eps}{\varepsilon}

\newcommand{\set}[1]{\left\{#1\right\}}

\DeclareMathOperator{\poly}{poly}

\newcommand{\hide}[1]{}

\DeclareMathOperator{\vol}{vol}
\DeclareMathOperator{\out}{out}

\begin{document}

%	\newgeometry{margin=1in,bottom=0.2in}

\title{\bf Faster Deterministic Distributed Coloring\\
Through Recursive List Coloring}

%\begin{flushleft}

\medskip
\author{
Fabian Kuhn \\
University of Freiburg, Germany \\
kuhn@cs.uni-freiburg.de}

\date{}
\maketitle

%\bigskip
\begin{abstract}
  We provide novel deterministic distributed vertex coloring
  algorithms. As our main result, we give a deterministic distributed
  algorithm to compute a $(\Delta+1)$-coloring of an $n$-node graph
  with maximum degree $\Delta$ in
  $2^{O(\sqrt{\log \Delta})}\cdot\log n$ rounds. This improves on the
  best previously known time complexity of
  $\min\set{O\big(\sqrt{\Delta\log\Delta}\log^*\Delta+\log^* n\big), 2^{O(\sqrt{\log
        n})}}$
  for a large range of values of $\Delta$. For graphs with arboricity
  $a$, we obtain a deterministic distributed algorithm to compute a
  $(2+o(1))a$-coloring in time $2^{O(\sqrt{\log a})}\cdot\log^2 n$.
  Further, for graphs with bounded neighborhood independence, we show
  that a $(\Delta+1)$-coloring can be computed more efficiently in time
  $2^{O(\sqrt{\log\Delta})} + O(\log^* n)$. This in particular implies
  that also a $(2\Delta-1)$-edge coloring can be computed
  deterministically in $2^{O(\sqrt{\log \Delta})} + O(\log^* n)$
  rounds, which improves the best known time bound for small values of
  $\Delta$.  All results even hold for the list coloring variants of
  the problems. As a consequence, we also obtain an improved
  deterministic $2^{O(\sqrt{\log\Delta})}\cdot\log^3 n$-round
  algorithm for $\Delta$-coloring non-complete graphs with maximum
  degree $\Delta\geq 3$. Most of our algorithms only require messages
  of $O(\log n)$ bits (including the $(\Delta+1)$-vertex coloring
  algorithms).

  Our main technical contribution is a recursive deterministic
  distributed list coloring algorithm to solve list coloring
  problems with lists of size $\Delta^{1+o(1)}$. Given some list coloring problem and an orientation of the
  edges, we show how to recursively divide the global color space
  into smaller subspaces, assign one of the subspaces to each node of
  the graph, and compute a new edge orientation such that for each
  node, the list size to out-degree ratio degrades at most by a
  constant factor on each recursion level.
\end{abstract}

\section{Introduction and Related Work}
\label{sec:intro}

Distributed graph coloring is one of the core problems and probably
also the most extensively studied problem in the area of
\emph{distributed graph algorithms}. It is one of the classic problems
to study the more general fundamental question of how to efficiently
break symmetries in a network in a distributed way. The network is
modeled as a $n$-node graph $G=(V,E)$ in which the nodes communicate over
the edges and the objective is to compute a proper
coloring of $G$. If the maximum degree of $G$ is $\Delta$, the goal
typically is to compute a coloring with $f(\Delta)$ colors for some
function $f$, where the standard task is to color $G$ with $\Delta+1$
colors.\footnote{$\Delta+1$ colors is what can be guaranteed by a
  simple sequential greedy algorithm. Consequently and maybe even more
  importantly in the context of distributed algorithms, it is the
  minimum number of colors necessary such that an arbitrary partial
  proper coloring can be extend to a proper color of all nodes of
  $G$.} The problem is usually studied in a synchronous message
passing model known as the \LOCAL model \cite{linial92,peleg00}. Each
node $v\in V$ is equipped with a unique $O(\log n)$-bit identifier
$\ID(v)$. The nodes operate in synchronous rounds, where in each round
each node can send a message to each of its neighbors.  Initially, the
nodes only know the IDs of their neighbors and at the end, each node
needs to know its own color of the computed vertex coloring of
$G$. The time complexity of an algorithm is measured as the number of
rounds. In the \LOCAL model, the maximum message size is not
bounded. The more realistic variant of the model where messages have
to be of size at most $O(\log n)$ bits is known as the \CONGEST model
\cite{peleg00}.

\paragraph{Early Work on Distributed Coloring:} The work on
distributed coloring started more than 30 years ago with the early
work on solving symmetry breaking problems in the parallel setting
\cite{alon86,cole86,goldberg88,luby86} and most importantly with a
seminal paper by Linial \cite{linial87}, which introduced the \LOCAL
model for solving graph problems in a distributed setting. Linial in
particular showed that a graph can be colored (deterministically) in
$O(\log^* n)$ rounds with $O(\Delta^2)$ colors and that this is tight
even for graph of degree $2$: Every
deterministic algorithm to color an $n$-node ring with a
constant number of colors requires $\Omega(\log^* n)$
rounds.\footnote{The $\log^* n$ function measures the number of times
  the $\log(\cdot)$ function has to applied iteratively in order to obtain a
  constant value when starting with value $n$.} Naor later extended the lower
bound to randomized algorithms~\cite{naor91}. To this day, the
$\Omega(\log^* n)$ lower bound is still the only known time lower
bound for $(\Delta+1)$-coloring in the \LOCAL model. Together with the
observation that as long as there are more than $\Delta+1$
colors, one can reduce the number of colors by $1$ in a single round
(by recoloring the nodes of the largest color), the $O(\Delta^2)$-coloring
algorithm of Linial directly implies a deterministic
$O(\Delta^2 + \log^* n)$-time algorithm for $(\Delta+1)$-coloring. It
was later shown that a refined single-round color reduction
scheme improves this running time to
$O(\Delta\log\Delta + \log^* n)$~\cite{szegedy93,Kuhn2006On}. 

If randomization is allowed, already the early works imply that a
$(\Delta+1)$-coloring can always be computed in $O(\log n)$
rounds. The result in particular follows from a simple reduction from
$(\Delta+1)$-coloring to the problem of computing a maximal
independent set (MIS) \cite{luby86,linial87} and from the classic
$O(\log n)$-time randomized algorithms for computing an MIS
\cite{alon86,luby86}.

\paragraph{Exponential Separation Between Randomized and Deterministic
  Complexity:} While there are very efficient (and simple) randomized
algorithms for solving $(\Delta+1)$-coloring, the best known
deterministic algorithm is almost exponentially slower an requires
time $2^{O(\sqrt{\log n})}$~\cite{panconesi95}. The algorithm is based
on a structure known as a network decomposition, which was introduced
by Awerbuch et al.~\cite{awerbuch89} and decomposes the graph into
clusters of small diameter and a coloring of the cluster graph with a
small number of colors. The distributed coloring problem is then
solved in a brute-force way: One iterates over the cluster colors and
for each cluster color, the respective clusters are colored by locally
learning the complete cluster topology in time proportional to the cluster
diameter. The question of whether there really is an exponential gap
between the best randomized and deterministic time complexities for
the $(\Delta+1)$-coloring problem (and in fact for many more problems)
is one of the most fundamental and long-standing open problems in the
area of distributed graph algorithms (see, e.g.,
\cite{linial87,barenboimelkin_book,stoc17_complexity}). In fact, in
their book on distributed graph coloring, Barenboim and
Elkin~\cite{barenboimelkin_book} explicitly state that \emph{``perhaps
  the most fundamental open problem in this field is to understand the
  power and limitations of randomization.''}. In the present paper, we
make progress on establishing the deterministic complexity of the
distributed $(\Delta+1)$-coloring problem. While we do not improve the
best time bound for general graphs, we significantly improve the best known
bound for a large range of values for $\Delta$. We also give the first
$2^{O(\sqrt{\log n})}$-time deterministic $(\Delta+1)$-coloring
algorithm that does not depend on computing a network
decomposition. In fact, our algorithm even works in the \CONGEST
model.

\paragraph{Deterministic Distributed Coloring:} As we focus on
deterministic distributed coloring algorithms, we first
discuss this case in more detail. While as a function of $n$, the
$2^{O(\sqrt{\log n})}$-time solution of \cite{panconesi95} is still
the fastest known deterministic algorithm, there has been extensive
work on determining the round complexity of distributed coloring as a
function of the maximum degree $\Delta$.  The first improvement over
the $O(\Delta\log\Delta + \log^*n)$-algorithm mentioned above was
obtained independently by Barenboim and Elkin \cite{barenboim09} and
Kuhn \cite{spaa09}, who showed that the problem can be solved in
time $O(\Delta + \log^* n)$. The algorithms of
\cite{barenboim09,spaa09} introduced \emph{distributed defective
  coloring} as a new technique. A $d$-defective $c$-coloring of a
graph is a coloring with $c$ colors such that the subgraph induced by
each color class has maximum degree at most $d$. Such a coloring
allows to solve the distributed coloring problem in a
divide-and-conquer fashion by decomposing the graph into several
smaller-degree subgraphs. This basic idea was further developed by
Barenboim and Elkin in \cite{barenboim10}, where they showed that for
every constant $\eps>0$, it is possible to compute an
$O(\Delta)$-coloring in time $O(\Delta^\eps\log n)$ and that a
$\Delta^{1+\eps}$-coloring can be computed in time
$O(\poly\log\Delta\cdot \log n)$. The paper introduced \emph{arbdefective
colorings} as a generalization of defective colorings. Such colorings
partition the nodes of a graph into subgraphs of small
arboricity\footnote{The arboricity of a graph $G$ is the minimum
  number of forests into which the edge set of $G$ can be decomposed.}
rather into subgraphs of small degree. Note that while the paper shows
that there are very efficient coloring algorithms that use a
relatively small number of colors, it does not improve the time
required to compute a $(\Delta+1)$-coloring. A further improvement for
the
$(\Delta+1)$-coloring problem was obtained by Barenboim
\cite{barenboim15}, who improves the time from
$O(\Delta+\log^* n)$ to $O(\Delta^{3/4}\log\Delta + \log^* n)$. The
key insight of the paper was to consider the more general list
coloring problem rather than the ordinary vertex coloring problem. In
a $C$-list coloring problem, every node $v\in V$ has to choose a color
from a list $L_v$ of size $|L_v|\geq C$, which is given to $v$
initially. List colorings in particular allow to extend previously
computed partial colorings.  Considering list colorings thus allows to
recursively decompose the graph and to then color the parts
sequentially without having to divide the color space among the
different parts. The result of \cite{barenboim15} was improved to
$O(\sqrt{\Delta}\cdot \poly\log\Delta + \log^* n)$ by Fraigniaud,
Heinrich, and Kosowski~\cite{fraigniaud16}, who introduced a
generalized version of list coloring, which they could solve more
efficiently. Our algorithms build on the techniques developed in all
the above works. From a technical point of view, our main contribution
is an extension of the recursive algorithm of \cite{barenboim10} to
list colorings. We then combine this list coloring algorithm
with the framework developed in \cite{barenboim15,fraigniaud16}.

We conclude the discussion on existing deterministic distributed
coloring by mentioning a few additional notable results. While except
for the original $\Omega(\log^* n)$ bound, no lower bound for the
\LOCAL model has been proven, it was recently shown that in a weak
version of the \LOCAL model (called the \SETLOCAL model), the
$(\Delta+1)$-coloring problem has a lower bound of
$\Omega(\Delta^{1/3})$ \cite{disc16_coloring}. Further, in
\cite{barenboim18}, an alternative and particularly elegant
deterministic $O(\Delta+\log^* n)$-time $(\Delta+1)$-coloring
algorithm was given. The algorithm reduces the colors iteratively in a
round-by-round manner rather than by using defective colorings. The
technique also allows to slightly improve the best
$(\Delta+1)$-coloring bound of \cite{fraigniaud16} to
$O(\sqrt{\Delta\log\Delta}\log^*\Delta + \log^* n)$~\cite{barenboim18}.  Finally,
some authors also considered the problem of coloring the nodes of a
graph with $\Delta$ instead of $\Delta+1$ colors. Although this
problem has a less local nature, it has been shown that up to
$\poly\log n$ factors, $\Delta$-colorings can essentially be computed
as fast as a
$(\Delta+1)$-colorings~\cite{panconesi95delta,podc18_Deltacoloring}.

\paragraph{Randomized Distributed Coloring:} There has also been a lot
of work on understanding the randomized complexity of the distributed
coloring problem
\cite{kothapalli06,SchneiderW10,barenboim12,PettieS13,HarrisSS16,podc18_Deltacoloring,chang18_coloring}. A
particularly important contribution was provided by Barenboim, Elkin,
Pettie, and Schneider \cite{barenboim12}, who introduced the so-called
graph shattering technique to the theory of distributed graph
algorithms. The paper shows that a $(\Delta+1)$-coloring can be
computed in time $O(\log\Delta) + 2^{O(\sqrt{\log\log n})}$. The basic
idea of shattering is to first use an efficient randomized algorithm
that solves the problem on most of the graph, has a time complexity
that often just depends on $\Delta$, and that
essentially only leaves unsolved components of size $\poly\log n$. In
the case of \cite{barenboim12}, this randomized part of the algorithm
requires $O(\log\Delta)$ rounds. The remaining problems on
$\poly\log n$-size components are then solved by using the most
efficient deterministic algorithm for the problem, i.e., in time
$2^{O(\sqrt{\log\log n})}$ in the case of distributed coloring. In
fact, in \cite{chang16}, Chang, Kopelowitz, and Pettie showed that
this shattering technique is essentially necessary and that the
randomized complexity of coloring (and many more graph problems) on
graphs of size $n$ is at least the deterministic complexity of the problem on
graphs of size $\sqrt{\log n}$. Note that this adds another very
strong motivation to understand the deterministic complexity of the
distributed coloring problem. In subsequent work, the time bound of
the randomized $(\Delta+1)$-coloring problem was improved to
$O(\sqrt{\log\Delta}) + 2^{O(\sqrt{\log\log n})}$ in \cite{HarrisSS16}
and finally to $2^{O(\sqrt{\log\log n})}$ in
\cite{chang18_coloring}. Note that in \cite{chang18_coloring}, the
time is dominated by the deterministic time for coloring components of
size $\poly\log n$. Hence, the current best randomized time complexity
of the $(\Delta+1)$-coloring problem can only be improved by at the
same time also improving the best deterministic algorithm for the
problem.

\paragraph{Distributed Edge Coloring:}
A particularly extensively studied special case of distributed
coloring is the edge coloring problem. The task corresponding to
computing a $(\Delta+1)$-vertex coloring is to compute an edge
coloring with $2\Delta-1$ colors.  It was already realized 20 years
ago that the distributed edge coloring problem has a simpler structure
than the distributed vertex coloring problem and there have been many
contributions towards the randomized
\cite{panconesi1997randomized,dubhashi1998near,ElkinPS15,ChangHLPU18}
and the deterministic
\cite{PanconesiR01,Czygrinow01,ghaffari17,FischerGK17,derandomization,stoc18_edgecoloring,Harris18,SuVu19}
complexity of the problem. In both case, we now know efficient
algorithms for computing $(2\Delta-1)$-edge colorings and (under some
conditions in the deterministic case) even for
computing $(1+\eps)\Delta$-edge colorings. Note that by Vizing's
classic theorem, every graph has an edge coloring with
$\Delta+1$ colors \cite{vizing1964estimate}.

A particularly important step was obtained by Fischer et
al.~\cite{FischerGK17}, who show that the $(2\Delta-1)$-edge coloring
problem can be solved deterministically in
$O(\poly\log\Delta \cdot \log n)$ rounds. Previously, the best
deterministic algorithm for general graphs had a time complexity of
$2^{O(\sqrt{\log n})}$ and as in the case of the best $(\Delta+1)$-vertex coloring
algorithm, it was a brute-force algorithm based on first computing a
network decomposition \cite{awerbuch89,panconesi95}. Together with the
best known randomized algorithms (which use the shattering technique),
the algorithm of \cite{FischerGK17} also implied that the
$(2\Delta-1)$-edge coloring problem has a randomized complexity of
$\poly\log\log n$. The current best time complexities for computing an
$(2\Delta-1)$-edge coloring in general graphs are
$\tilde{O}(\log^2\Delta\log n)$ for the deterministic and
$\tilde{O}(\log^3\log n)$ for the randomized setting
\cite{Harris18}. For small values of $\Delta$, the best known
complexity is
$O(\sqrt{\Delta\log\Delta}\log^*\Delta +
\log^*n)$~\cite{fraigniaud16,barenboim18},
as for computing a $(\Delta+1)$-vertex coloring. As one of our
contributions, we improve this last result to
$2^{O(\sqrt{\log\Delta})}+O(\log^* n)$.

\subsection{Contributions and Main Ideas}
\label{sec:contributions}

We will now describe our contributions and the main ideas that
lead to those contributions. 

\paragraph{Recursive Distributed List Coloring:} As stated above, from
a technical point of view, our main contribution is a recursive solution
for the list coloring problem. Our algorithm generalizes the algorithm
of Barenboim and Elkin in \cite{barenboim10}. In
\cite{barenboim10}, the node set of $G$ is recursively divided into
smaller sets and the subgraphs of $G$ induced by 
these node sets are then colored independently by using disjoint sets of colors for
the different parts. The recursive partitioning is done by a so-called
arbdefective coloring. Given a graph $G$ with an edge orientation
where the out-degree of each node is at most $\beta$, a
$\beta'$-arbdefective coloring of $G$ is a coloring and a new
orientation of $G$ such that each node has at most $\beta'$
out-neighbors (w.r.t.\ to the new orientation) of its own color. In
\cite{barenboim10}, it is shown that in time $O(k^2\log n)$, it is
possible to compute an $O(\beta/k)$-arbdefective coloring with $k$
colors. Note that this recursive partitioning of the graphs requires
the total number of colors to be sufficiently large and it in
particular cannot provide a $(\Delta+1)$-coloring. In each recursive
step, the number of parts grows by a constant factor more than the maximum
out-degree is reduced.

When trying to extend the above solution to list coloring, one could
try to divide the space of all possible colors into $k$ parts and
assign each node to one of the parts. If we could find such a
partition such that the color list $L_v$ of each node is roughly
divided by a factor $k$, we could use the same arbdefective colorings as
in \cite{barenboim10} to recursively divide problem. Finding such a
partition of the color space is however a problem that does not have a
local solution (whether such a partition is good depends on all the
lists). We therefore adapt the solution of \cite{barenboim10} in a
more careful way that is specifically designed to work for the list
coloring problem. We consider an arbitrary partition of the global
color space into $k$ (roughly) equal parts and we partition the nodes
by assigning each node to one of the $k$ parts. However, instead of
equally reducing the out-degrees, we find a partition and a new
orientation such that for every node $v$, the size of its list gets
smaller at roughly the same rate as the out-degree of $v$ (again up to
a constant factor). For a formal definition of our generalization of
arbdefective colorings, we refer to \Cref{def:orientedreduction}. If
the initial lists are sufficiently large, this recursive partitioning
allows to maintain list coloring problems that are solvable. We
formally prove the following main technical result. The construction
follows the same basic ideas as the construction in
\cite{barenboim10}, carefully adapted to the more general list
coloring setting.

\begin{theorem}
  \label{thm:mainrecursive}
  Let $G=(V,E)$ be a graph where all edges are oriented and for each
  $v\in V$, let $\beta(v)$ be the out-degree of $v$. Let
  $\eps\in(0,1]$ and $r\geq 1$ be two parameters and assume that each
  node $v$ is assigned a list $L_v$ of size at least
  $|L_v|>(2+\eps)^r\beta(v)$ consisting of colors from a color space
  of size at most $C$. Then, there is a deterministic
  $O\big(\frac{r\cdot C^{2/r}}{\eps^3}\cdot\log n\big)$-round
  algorithm to solve the given list coloring instance. The algorithm
  requires messages of size $O(\log n + \log C)$.
\end{theorem}

Note that \Cref{thm:mainrecursive} directly gives a generalization of
the results of \cite{barenboim10} to list colorings. Specifically, for
an arbitrary initial edge orientation, the theorem implies that if $C$
is a most polynomial in $\Delta$, for every constant $\eps>0$, there
are deterministic distributed coloring algorithms to solve list
coloring instances with lists of size $c\Delta$ for a sufficiently
large constant $c>1$ in
$O(\Delta^\eps\log n)$ rounds and list coloring instances with lists
of size $\Delta^{1+\eps}$ in $O(\poly\log\Delta\cdot\log n)$ rounds. The
two results are obtained by setting $r=\Theta(1/\eps)$ and
$r=\Theta(\eps\log\Delta)$, respectively.

\paragraph{\boldmath$(\Delta+1)$-List Coloring:} By using the
framework developed in \cite{barenboim15,fraigniaud16}, we can use the
list coloring algorithm of \Cref{thm:mainrecursive} to efficiently
compute list colorings for lists of size $\Delta+1$. The basic idea is
to first use a defective coloring to partition the nodes of $G$ into
$\Delta^{O(\eps)}$ parts of maximum degree at most
$\Delta^{1-\eps}$. The algorithm then sequentially iterates through
the $\Delta^{O(\eps)}$ parts and when processing part $i$, it colors
the nodes of that part so that the coloring is consistent with the
coloring computed for the previous parts. Note that when coloring a
part in a consistent way with already colored parts, each node $v$ is
only allowed to choose a color that has not already been chosen by a
neighbor of $v$. Coloring each part is therefore a list coloring
problem, even if the original problem is not a list coloring
problem. As long as at least half the neighbors of a node $v$ are
uncolored when processing node $v$, the list of $v$ is of length more
than $\Delta/2$. Because each part has maximum degree at most
$\Delta^{1-\eps}$, we are thus in a setting where we can apply
\Cref{thm:mainrecursive}. After going through the $\Delta^{O(\eps)}$
parts, we get a partial list coloring where the maximum degree of the
graph induced by the uncolored nodes is at most $\Delta/2$. Iterating
$O(\log\Delta)$ times then gives a $(\Delta+1)$-coloring.  By setting
$\eps=\Theta(1/\sqrt{\log\Delta})$, we obtain our main theorem, which
is stated in the following.

\begin{theorem}\label{thm:generalDelta}
  Let $G=(V,E)$ be an $n$-node graph with maximum degree $\Delta$ and
  assume that we are given a $(\deg(v)+1)$-list coloring instance with
  colors from a color space of size $\poly(\Delta)$. Then, there is a
  deterministic \CONGEST model algorithm to solve this list coloring
  problem in time $2^{O(\sqrt{\log \Delta})} \log n$.
\end{theorem}

Note that by $(\deg(v)+1)$-list coloring we refer to a list coloring
instance where the size of the list of each node $v$ is at least
$\deg(v)+1$. The theorem thus in particular implies that the standard
$(\Delta+1)$-coloring problem can be solved in time
$2^{O(\sqrt{\log\Delta})}\log n$ in the \CONGEST
model. \Cref{thm:generalDelta} improves the best known deterministic
time complexity of
$\min\set{\tilde{O}(\sqrt{\Delta}) + O(\log^*n), 2^{O(\sqrt{\log
      n})}}$
for solving $(\Delta+1)$-coloring even in the \LOCAL model for all
$\Delta$ between $\log^{2+\eps}n$ and $n^{o(1)}$. The theorem also
extends the range of $\Delta$ for with $(\Delta+1)$-coloring can be
solved in deterministic $\poly\log n$ time from
$\Delta\leq\poly\log n = 2^{O(\log\log n)}$ to
$\Delta = 2^{O(\log^2\log n)}$. It is further notable that our
algorithm gives an alternative way of solving $(\Delta+1)$-coloring in
deterministic $2^{O(\sqrt{\log n})}$ time. Our algorithm does not
depend on computing a network decomposition and solving the clusters
in a brute-force way and the algorithm even works in the \CONGEST
model (i.e., with $O(\log n)$-bit messages).

\paragraph{Bounded Neighborhood Independence:} A graph $G=(V,E)$ is
said to have neigborhood independence $\theta$ if every node $v\in V$
has at most $\theta$ independent (i.e., pairwise non-adjacent)
neighbors (see also \Cref{sec:prelim}). In \cite{barenboim11},
Barenboim and Elkin showed that in graphs of bounded neighborhood
independence, the recursive coloring algorithm of \cite{barenboim10}
can be improved: Essentially, they show that the multiplicative
$O(\log n)$ factor in all the time bounds of \cite{barenboim10} can be
replaced by an additive $O(\log^* n)$ term if the neighborhood
independence is $\theta=O(1)$. We show that the same improvement can
also be obtained for list colorings, leading to the following theorem.

\begin{theorem}\label{thm:BNI}
  Let $G=(V,E)$ be an $n$-node graph with maximum degree $\Delta$ and
  neighborhood independence at most $\theta\geq 1$. Assume that we are
  given a $(\deg(v)+1)$-list coloring instance with colors from a
  color space of size $\poly(\Delta)$. Then, there is a deterministic
  \CONGEST model algorithm to solve this list coloring problem in time
  $2^{O(\sqrt{\log\theta\cdot \log \Delta})} +O(\log^* n)$.
\end{theorem}

The line graph of any graph $G$ has neighborhood independence at most
$2$ (for each edge $\set{u,v}$ of $G$, the adjacent edges can be
covered by two cliques in the line graph, defined by the edges
incident to $u$ and the edges incident to $v$). \Cref{thm:BNI} thus leads to
the following direct corollary for the
\emph{distributed edge coloring} problem.

\begin{corollary}\label{cor:edgecoloring}
  Let $G=(V,E)$ be an $n$-node graph with maximum degree $\Delta$ and
  assume that we are given a lisnt edge coloring instance with
  lists from a color space of size at most $\poly(\Delta)$. If the
  list $L_e$ of each edge $e=\set{u,v}\in E$ is of size at least
  $\deg(u)+\deg(v)-1$, there is a deterministic \LOCAL algorithm
  to solve the given list edge coloring problem in time
  $2^{O(\sqrt{\log\Delta})} + O(\log^* n)$.
\end{corollary}

Note that even though on the line graph, the algorithm of
\Cref{thm:BNI} can be implemented in the \CONGEST model (i.e., with
messages of size $O(\log n)$ bits), it is not immediate whether the
same is also true when implementing the algorithm on $G$. The theorem
therefore only implies a \LOCAL algorithm for the edge coloring
problem.\footnote{We believe that it might be possible to actually
  also implement the algorithm in the \CONGEST model on $G$.}
The result in particular improves the time for computing a
$(2\Delta-1)$-edge coloring for $\Delta\leq \log^{2-\eps} n$ from
$\tilde{O}(\sqrt{\Delta}) + O(\log^* n)$ to $2^{O(\sqrt{\log\Delta})}
+ O(\log^* n)$.

\paragraph{Applications:}

\Cref{thm:generalDelta} has two immediate applications. If $G=(V,E)$
is a graph with arboricity $a$, one can use an algorithm from
\cite{barenboim08} to decompose $V$ into $p=O(\log(n)/\eps)$ parts
$V_1,V_2,\dots,V_{p}$ such that every node $v\in V_i$ ($\forall i$)
has at most $(2+\eps)a$ neighbors in $\bigcup_{j=i}^p V_j$. The
time for computing this decomposition is $O(\log(n)/\eps)$
(deterministically in the \CONGEST model). By sequentially list
coloring the sets $V_1,\dots,V_p$ in reverse order, we get the
following corollary of \Cref{thm:generalDelta}.

\begin{corollary}\label{cor:generalArb}
  Let $G=(V,E)$ be an $n$-node graph with arboricity $a$ and assume
  that for some $\eps>0$, we are given a $\lceil(2+\eps)a\rceil$-list
  coloring instance with lists from a color space of size at most
  $\poly(\Delta)$. Then there is a deterministic \CONGEST model
  algorithm that solves this list coloring problem in time
  $\frac{1}{\eps}\cdot 2^{O(\sqrt{\log a})}\cdot \log^2 n$.
\end{corollary}

We note that the number of colors is close to optimal as there are
graphs with arboricity $a$ and chromatic number $2a$ (see also the
more detailed discussion in \Cref{sec:prelim}).

Finally, we also obtain an improved algorithm for deterministically
computing a $\Delta$-coloring of a graph $G$. By a classic result of
Brooks from 1941~\cite{brooks2009colouring}, it is well-known that
every graph $G$ with maximum degree $\Delta$ can be colored with
$\Delta$ colors unless $G$ is an odd cycle or a complete graph. In
\cite{podc18_Deltacoloring}, it was shown that for $\Delta\geq 3$, in
the \LOCAL model, the problem of computing a $\Delta$-coloring can be
deterministically reduced to $O(\log^2 n)$ instances of computing a
$(\deg(v)+1)$-list coloring problem. This implies the following
corollary of \Cref{thm:generalDelta}.

\begin{corollary}\label{cor:Deltacoloring}
  Let $G=(V,E)$ be an $n$-node graph with maximum degree $\Delta$. If
  $G$ is not a complete graph, a $\max\set{3,\Delta}$-coloring of $G$ can
  be computed deterministically in $2^{O(\sqrt{\log\Delta})}\cdot \log^3
  n$ rounds in the \LOCAL model.
\end{corollary}

%%% Local Variables:
%%% mode: latex
%%% TeX-master: "main"
%%% End:

\section{Model and Preliminaries}
\label{sec:prelim}

\paragraph{Mathematical Preliminaries:}
Let $G=(V,E)$ be an undirected graph. For a node $v\in V$, we use $N_G(v):=\set{u\in V : \set{u,v}\in E}$ to denote the set of neighbors and we use $\deg_G(v):=|N(v)|$ to denote the degree of $v$ (if it is clear from the context, we omit the subscript $G$). The maximum degree of $G$ is denoted by $\Delta(G)$ and the arboricity of $G$ (definition see below) is denoted by $a(G)$. Again, when $G$ is clear from the context, we just use $\Delta$ and $a$ instead of $\Delta(G)$ and $a(G)$. The \emph{neighborhood independence} of a graph $G=(V,E)$ is the size of the largest independent set of any subgraph $G[N(v)]$ that is induced by the neighbors $N(v)$ of some node $v\in V$. For any natural number $k\in \mathbb{N}$, we use $[k]$ to denote the set $[k] :=\set{1,\dots,k}$. If not specified otherwise, logarithms are to base $2$, i.e., $\log x$ means $\log_2 x$.

\paragraph{Communication Model:}
We assume the standard synchronous communication model in graphs. The network is modeled as an $n$-node graph $G=(V,E)$, where each node hosts a distributed process (which for simplicity will be identified with the node itself). Time is divided into synchronous rounds. In each round, each node $v\in V$ can send a possibly different message to each of its neighbors in $G$, receive the messages sent to it by its neighbors in the current round, and perform some arbitrary internal computation. We assume that each node $v\in V$ has a unique $O(\log n)$-bit identifier $\ID(v)$. Depending on how much information can be sent over each edge of $G$ in a single round, one usually distinguishes two variants of this model. In the \LOCAL model, the messages can be of arbitrary size and in the \CONGEST model, messages have to be of size at most $O(\log n)$ bits. To keep the arguments as simple as possible, we assume that all nodes know the number of nodes $n$ and the maximum degree $\Delta$ of $G$. For algorithms that depend on the arboricity $a$ of $G$, we also assume that all nodes know $a$. All our results are relatively straightforward to generalize to the case where the nodes only know a polynomial upper bound on $n$ and linear upper bounds on $\Delta$ and $a$.\footnote{If the IDs are from a possibly larger space $1,\dots,N$, all results still directly hold with an additional additive time cost of $O(\log^* N)$ as long as we are allowed to also increase the maximum message size to $O(\log N)$ (in order to fit a single ID in a message). Further, with standard ``exponential guessing'' techniques, it should be possible to get rid of the knowledge of $\Delta$ and $a$ entirely (see, e.g., \cite{korman11}).}

\paragraph{Arboricity:}
The \emph{arboricity} $a$ of a graph $G=(V,E)$ is the minimum number of forests that are needed to cover all edges of $E$. This implies that the number of edges is at most $a(n-1)$ and more generally, any $k$-node subgraph has at most $a(k-1)$ edges. The arboricity is thus a measure of the sparsity of $G$. A classic theorem of Nash-Williams shows that the arboricity is equal to the smallest integer $a$ such that every subgraph $G_S=(V_S,E_S)$ of $G$ has at most $a(|V_S|-1)$ edges \cite{nash-williams64}. 

Note that every subgraph of a graph $G$ of arboricity $a$ has a node of degree strictly less than $2a$. By iteratively removing a node of degree at most $2a-1$ and coloring the nodes in the reverse order, we can therefore always find a vertex coloring with at most $2a$ colors. Note that this is tight in general as Nash-William's theorem for example directly implies that the complete graph $K_{2a}$ of size $2a$ has arboricity $a$. The process of iteratively removing a node of degree at most $2a-1$ until the graph is empty also implies that every graph of arboricity $a$ has an acyclic orientation of the edges such that each node has out-degree at most $2a-1$. If we drop the requirement that the orientation needs to be acyclic, the edges can be oriented such that the out-degree of each node is at most $a$ (by taking an out-degree $1$ orientation of each of the $a$ forests into which the edge set $E$ decomposes). 

\paragraph{List Coloring:}
In a \emph{list coloring} instance of a graph $G=(V,E)$, every node $v\in V$ is initially assigned a list $L_v$ of colors and the objective is to assign a color $x_v\in L_v$ such that the assigned colors form a proper vertex coloring of $G$. A graph $G$ is said to be $c(v)$-list-colorable if any assignment of lists of size $|L_v|\geq c(v)$ for each $v\in V$ has a solution. Note that every graph is $(\deg(v)+1)$-list colorable and every graph of arboricity $a$ is $2a$-list colorable. Some of our algorithms only compute partial colorings. For a given list coloring instance with lists $L_v$, a partial proper list coloring is an assignment of values $x_v\in\set{\bot}\cup L_v$ to each node such that for any two neighbors $u$ and $v$, $x_u\neq x_v$ or $x_u=x_v=\bot$ (where an assignment of $\bot$ is interpreted as not assigning a color to a node).

% \paragraph{Basic Distributed Coloring Algorithms:}

% The following theorem summarizes some basic known deterministic distributed coloring algorithms that we will use in our algorithms.

% \begin{theorem}\cite{linial92,barenboim08,BEK15}
%   Let $G=(V,E)$ be an $n$-node graph with maximum degree $\Delta$ and arboricity $a$.
%   \begin{compactenum}[(I)]
%   \item There is a deterministic  $O(\log^* n)$-time \CONGEST algorithm to color $G$ with $O(\Delta^2)$ colors.
%   \item Given an $O(\Delta^2)$-coloring of $G$, there is a deterministic $O(\Delta)$-time \CONGEST algorithm to solve an arbitrary $(\deg(v)+1)$-list coloring instance of $G$.
%   \item For any $\eps>0$, there is a deterministic $O\big(\frac{a\cdot \log n}{\eps}\big)$-time \CONGEST algorithm to solve an arbitrary $\lceil(2+\eps)a\rceil$-list coloring instance of $G$.
%   \end{compactenum}
% \end{theorem}

% Result (I) was shown by Linial in \cite{linial92}, result (II) was shown by Barenboim et al.\ in \cite{BEK15}, and result (III) was shown by Barenboim and Elkin in \cite{barenboim08}. We note that there are faster algorithms for solving $(\deg(v)+1)$-list coloring instances, which were proven in \cite{barenboim15,fraigniaud16,barenboim18} (see discussion in \Cref{sec:intro}), however the simpler linear-time algorithm of \cite{BEK15} is sufficient for our purpose. Note also that result (III) implies that it is possible to solve any $(2a+1)$-list coloring instance in time $O(a^2\log n)$ (by setting $\eps=1/(2a)$).

\paragraph{Defective Colorings:}

For a graph $G=(V,E)$ and two integers $d\geq 0$ and $c\geq 1$, a \emph{$d$-defective $c$-coloring} is an assignment of a color $x_v\in [c]$ to each node $v\in V$ such that for every $v\in V$, the number of neighbors $u\in N(v)$ with $x_u=x_v$ is at most $d$. We also use a version of defective coloring where the defect of a node $v$ can be a function of the degree $\deg(v)$. For any $\lambda\in [0,1]$ and an integer $c\geq 1$, we define a \emph{$\lambda$-relative defective $c$-coloring} of $G$ as an assignment of colors $x_v\in[c]$ such that each node $v\in V$ has at most $\lambda\cdot\deg(v)$ neighbors of the same color. We will make frequent use of the following distributed defective coloring result, which was proven in \cite{spaa09,relativedefect}.

\begin{theorem}\cite{spaa09,relativedefect}\label{thm:basicdefective}
  Let $G=(V,E)$ be a graph, $\lambda>0$ a parameter, and assume that an $O(\poly(\Delta))$-coloring of $G$ is given. Then, there is a deterministic $O(\log^* \Delta)$-time \CONGEST algorithm to compute an $\lambda$-relative defective $O(1/\lambda^2)$-coloring of $G$.
\end{theorem}

By using an algorithm of Linial \cite{linial92}, it is possible to deterministically compute an $O(\Delta^2)$-coloring of $G$ in $O(\log^* n)$ rounds in the \CONGEST model. As all our time complexities are lower bounded by $\Omega(\log^* n)$, we will generally assume that an $O(\Delta^2)$-coloring of the network graph $G$ is given. In \cite{spaa09}, \Cref{thm:basicdefective} was proven for the standard defective coloring definition. In \cite{relativedefect}, it was observed and formally proven that the algorithm and proof of \cite{spaa09} directly extends to the relative defective coloring problem defined above.  In \cite{barenboim10}, Barenboim and Elkin defined a generalization of defective coloring that decomposes the graph into subgraphs of small arboricity (rather than into subgraphs of small degree). Formally, for integers $\beta\geq 0$ and $c\geq 1$, a \emph{$\beta$-arbdefective $c$-coloring} of $G$ is an assignment of colors $x_v\in[c]$ to the nodes in $v\in V$ and an orientation of the edges of $G$ such that each node $v\in V$ has at most $\beta$ out-neighbors of the same color.

\paragraph{List Color Space Reduction:}

Defective and arbdefective colorings are a natural way to decompose a given distributed coloring problem into multiple coloring problems on sparser (and thus easier to solve) instances. Over the last 10 years, the technique has been applied successfully in several distributed coloring algorithms \cite{barenboim09,spaa09,barenboim10,barenboim11,barenboim15,fraigniaud16}. In its most basic variant, the technique divides the color space into several disjoint subspaces and it then in parallel solves the resulting coloring problems for each of the smaller color spaces. The main technical contribution of this paper are generalizations of defective and arbdefective colorings that allow to reduce the color space for list colorings. For our main result, we need a generalization of a decomposition into parts of low arboricity, which is defined next. In \Cref{sec:neighborhoodindep}, we will also introduce a similar unoriented notion of list color space reduction.

\begin{definition}[Oriented List Color Space Reduction]\label{def:orientedreduction}
  Let $G=(V,E)$ be a graph where all edges are oriented. For each $v\in V$, let $\beta(v)$ be the out-degree of $v$. Assume that we are given a list coloring instance on $G$ with lists $L_v\subseteq \calC$ for some (global) color space $\calC$. Assume also that we are given a partition $\calC=\calC_1\cup\dots\cup\calC_p$ of the color space. Further, assume that each node $v\in V$ is assigned an index $i_v\in[p]$ and that we are given a new edge orientation $\pi$. For each $v\in V$, we define new lists $L_v':=L_v\cap \calC_{i_v}$ and we let $\beta'(v)$ be the number of out-neighbors $u$ for which $L_u'\cap L_v'\neq \emptyset$ (w.r.t.\ the new orientation $\pi$). The assignment $i_v$ together with the edge orientation $\pi$ is called an \emph{oriented $(\eta,\gamma)$-list color space reduction} for $\eta>1$ and $\gamma\geq 1$ if 
  \[
  \forall i\in[p]\,:\ |\calC_i| \leq \frac{1}{\eta}\cdot |\calC|
  \quad\text{and}\quad
  \forall v\in V\,:\ \frac{|L_v'|}{\beta'(v)} \geq \frac{1}{\gamma}\cdot \frac{|L_v|}{\beta(v)}.
  \]
\end{definition}

Hence, essentially, a list color space reduction allows to reduce the global color space by a factor $\eta$ while losing a factor of $\gamma$ in the list size to out-degree ratio. We will therefore try to find color space reductions where $\eta$ is as large as possible and $\gamma$ is as small as possible. When starting with a color space of size $C$ and applying $t=O(\log_\eta C)$ $(\eta,\gamma)$-list color space reductions, this will allow to compute a proper coloring as long as the initial list size to degree ratio is larger than $\gamma^t$.

%%% Local Variables:
%%% mode: latex
%%% TeX-master: "main"
%%% End:

\section{Distributed List Coloring in General Graphs}
\label{sec:generalgraphs}

In this section, we provide our list coloring algorithms for general
graphs. In \Cref{sec:recursive}, we describe the basic 
recursive algorithm that allows to prove
\Cref{thm:mainrecursive}. In \Cref{sec:iterativegeneral}, we show how
to combine the scheme with the framework of
\cite{barenboim15,fraigniaud16} to obtain \Cref{thm:generalDelta}, our
main result.

\subsection{Basic Recursive Distributed List Coloring}
\label{sec:recursive}

Our scheme is a generalization of a coloring algorithm by Barenboim
and Elkin in \cite{barenboim10}. The algorithm of \cite{barenboim10}
is based on recursively decomposing a given graph $G$ of arboricity
$a$ into $O(k)$ parts of arboricity $a/k$. While this allows to
recursively divide the color space of a standard vertex coloring
problem, it does not work for list coloring (it is not clear how to
evenly divide the color space such that also the lists of all the
nodes are divided in a similar way). We will show that by replacing
the arbdefective coloring of \cite{barenboim10} with an oriented list
color reduction, a similar algorithm works and it can also be used in
an analogous way recursively.  The decomposition of the nodes in
\cite{barenboim10} is based on first computing a so-called
$H$-partition \cite{barenboim08}, a partition of the nodes of $G$ into
$h=O(\log(n) / \eps)$ sets $V_1,\dots,V_h$ such that for each
$i\in [h]$ and every $v\in V_i$, $v$ has at most $(2+\eps)a$ neighbors
in set $V_j$ for $j\geq i$. We define a generalized $H$-partition as
follows.

\begin{definition}[Generalized $H$-Partition]\label{def:Hpartition}
  Let $G=(V,E)$ be a graph with a given edge orientation such that the
  out-degree of each node $v\in V$ is upper bounded by $\beta(v)$. A
  \emph{generalized $H$-partition} of $G$ with parameter $\alpha$ and
  of depth $h$ is a partition of $V$ into sets $V_1,\dots,V_{h}$
  such that for every $i\in [h]$ and every $v\in V_i$, the number
  of neighbors of $v$ in the set $\bigcup_{j=i}^h V_j$ is at most
  $\alpha\cdot\beta(v)$.
\end{definition}

The following lemma shows that every edge-oriented graph $G$ has a
generalized $H$-partition of logarithmic depth with a constant
parameter and that such a partition can also be computed efficiently
in the distributed setting. The construction is a natural adaptation
of the standard $H$-partition construction of \cite{barenboim08}.

\begin{lemma}\label{lemma:Hpartition}
  Let $G=(V,E)$ be a graph with a given edge orientation such that for
  each $v\in V$, $\beta(v)$ denotes the out-degree of $v$. For every
  $0<\eps=O(1)$, a generalized $H$-partition of $G$ of depth
  $h=O(\log(n) / \eps)$ and with parameter $2+\eps$ can be computed
  deterministically in $O(\log(n) / \eps)$ rounds in the \CONGEST
  model.
\end{lemma}
\begin{proof}
  We first define the construction of the partition
  $V_1,\dots,V_h$. In the following, we call the nodes $V_i$ the
  nodes of level $i$ of the construction. We iteratively construct the
  sets level by level. For each
  $i\in\set{0,\dots,h}$, let $U_i := V \setminus \bigcup_{j=1}^i
  V_i$ be the set of nodes that are not in the first $i$ levels and
  let $G_i := G[U_i]$ be the (oriented) subgraph of $G$ induced by
  $U_i$. We define the set $V_i$ to be the set of nodes $v\in U_{i-1}$
  that have degree at most $(2+\eps)\beta(v)$ in $G_i$. We stop the
  construction as soon as $U_h=\emptyset$. It is clear from the
  construction that each node $v\in V_i$ has at most
  $(2+\eps)\beta(v)$ neighbors in sets $V_i\cup\dots\cup V_h$. It is
  also straightforward to see that a single level can be constructed
  in $O(1)$ rounds in the \CONGEST model and the whole construction
  therefore requires $O(h)$ rounds. It therefore remains to show
  that the construction ends after constructing $O(\log(n) / \eps)$ levels.

  We consider one of the graphs $G_i=G[U_i]$ for $i\geq 0$ and for
  each $v\in U_i$, we define $\beta_{G_i}(v)\leq \beta(v)$ to be the
  out-degree of $v$ in $G_i$. We define
  $W :=\set{u\in U_i : \deg_{G_i}(u) \leq (2+\eps)\beta_{G_i}(u)}$ and
  $\overline{W} := U_i \setminus W$. Note that $W\subseteq V_{i+1}$
  (i.e., layer $i$ contains at least all nodes in $W$).  For a set
  $S\subseteq U_i$, we define $\vol(S):=\sum_{v\in S} \deg_{G_i}(v)$
  and $\out(S):=\sum_{v\in S}\beta_{G_i}(v)$. We will show that
  \begin{equation}
    \label{eq:volume}
    \out(W) > \eps\cdot\out(\overline{W}).
  \end{equation}
  \Cref{eq:volume} implies that the set $V_{i+1}$ contains at least an
  $\eps$-fraction of all the out-degrees of $G_i$. The sum of the
  out-degrees in $G_{i+1}$ is at most the sum of the out-degrees of
  the nodes $\overline{W}$ in $G_i$ and thus \Cref{eq:volume} implies that
  with each constructed level, the sum of the out-degrees decreases by
  at least a factor $1-\eps$. \Cref{eq:volume} therefore directly
  implies $h=O(\log(n) / \eps)$ and it thus remains to prove
  \Cref{eq:volume}. We have
  \[
  \vol(U_i) = \vol(W) + \vol(\overline{W}) = 2\out(W) + 2\out(\overline{W}).
  \]
  By the definition of $W$, we further have
  \[
  \vol(\overline{W}) > (2+\eps)\out(\overline{W}).
  \]
  Combining the two inequalities and the fact that for every $S\subseteq V$, $\out(S)\leq \vol(S)$, we have
  \begin{eqnarray*}
    \out(W)  \leq  \vol(W) & = & 2\out(W) + 2\out(\overline{W}) - \vol(\overline{W})\\
                           & < & 2\out(W) - \eps\out(\overline{W}),
  \end{eqnarray*}
  and thus \Cref{eq:volume} holds.
\end{proof}

Based on the generalized $H$-decomposition construction, we can now
prove the main technical lemma for general graphs, which proves the
existence of efficiently constructible oriented list color space
reductions.

\begin{lemma}\label{lemma:orientedreduction}
  Let $G=(V,E)$ be an edge-oriented graph, where the out-degree of
  each node $v\in V$ is denoted by $\beta(v)$. Let us further assume
  that each node $v\in V$ is assigned a list of colors $L_v$, where
  $L_v\subseteq \calC$ for some color space $\calC$ of size $C$. For every integer
  $\eta\in[1,C]$ and every $\eps\in(0,1]$, there is an oriented
  $(\eta,2+\eps)$-list color space reduction that can be computed
  deterministically in
  $O\big(\frac{\eta^2}{\eps^3}\cdot\log n\big)$ rounds with messages of
  size $O(\log n + \log C)$.
  % In the special case where $\eta=|\calC|$, the time complexity
  % can be improved to $O\big(\frac{\eta}{\eps}\cdot\log n\big)$.
\end{lemma}
\begin{proof}
  Let $\calC=\calC_1\cup\dots\cup\calC_{p}$ be an arbitrary partition
  of $\calC$ into $p=O(|\calC|/\eta)$ parts such that for all
  $x\in[p]$, $|\calC_i|\leq |\calC|/\eta$. In order to construct an
  oriented list color space reduction for the given
  partition of $\calC$, we need to provide an orientation $\pi$
  of the edges and an assignment of values $x_v\in[p]$ to all nodes
  $v\in V$. We define $L_v':=L_v\cap \calC_{x_v}$ as the new color
  list of $v$ and $\beta'(v)$ to be the number of out-neighbors of $v$
  (w.r.t.\ the new orientation $\pi$) for which $x_u=x_v$. To show
  that we have computed an oriented $(\eta,2+\eps)$-list color space
  reduction, we need to show that
  \begin{equation}\label{eq:listreduction1}
    \beta'(v) \leq
    (2+\eps)\cdot\frac{|L_v'|}{\alpha_v},\quad\text{where}\ 
    \alpha_v = \frac{|L_v|}{\beta(v)}.
  \end{equation}

  Let $\delta:=\eps/2$. As a first step, we compute a generalized
  $H$-partition $V=V_1\cup\dots\cup V_{h}$ with parameter $2+\delta$ and of depth $h=O(\log(n) /
  \delta)=O(\log(n) / \eps)$. By using \Cref{lemma:Hpartition}, this
  generalized $H$-partition can be computed in $O(\log(n) / \eps)$
  rounds in the \CONGEST model. We then consider the
  subgraphs $G_i=G[V_i]$ induced by all the levels $i\in [h]$ of the
  generalized $H$-partition. Note that each node $v\in V$ has degree
  at most $(2+\delta)\beta(v)$ in its subgraph $G_i$. For each of the
  graphs $G_i$, we (in parallel) compute a $(\delta/p)$-relative
  defective coloring. By applying \Cref{thm:basicdefective}, such a
  coloring with $O(p^2/\delta^2)=O(\eta^2/\eps^2)$ colors can be
  computed in $O(\log^* n)$ rounds. Based on the generalized
  $H$-partition and the defective coloring of each of the levels, we
  now define a partial orientation $\pi'$, which we will later extend
  to obtain the orientation $\pi$. Let $F\subseteq E$ be the set of
  edges $\set{u,v}$ of $G$ such that either $u$ and $v$ are in
  different levels of the partition or such that if $u$ and $v$ are 
  on the same level $i$, the defective coloring of $G_i$ assigns $u$
  and $v$ different colors. Orientation $\pi'$ orients all edges in
  $F$ and it leaves all edges in $E\setminus F$ unoriented. For
  $\set{u,v}\in F$, if the two nodes are in different levels $V_i$ and
  $V_j$ for $i<j$, the edge is oriented from $V_i$ to $V_j$ (i.e.,
  from the smaller to the larger level). If $u$ and $v$ are on the
  same level but have different colors, the edge is oriented from the
  larger to the smaller color. Note that by the property of the
  generalized $H$-partition, every node $v\in V$ has at most
  $(2+\delta)\beta(v)$ neighbors $u$ for which $\pi'$ orients $v$ to
  $u$. Further by the property of the defective colorings, every node
  $v\in V$ has at most $\delta/p\cdot \beta(v)$ neighbors to
  which the edge is not oriented by $\pi'$.

  We are now ready to assign the color subspaces
  $\calC_1,\dots,\calC_p$ to the nodes, i.e., to assign $x_v\in [p]$
  to each $v\in V$. We do this in $h$ phases, each consisting of
  $O(p^2/\delta^2)=O(\eta^2/\eps^2)$ rounds. In the phases, we iterate through the
  $h$ levels of the partition in reverse order (i.e., we first
  process the nodes in $V_h$, then in $V_{h-1}$, and so on). In
  the phase when processing level $i\in [h]$, we iterate through
  the $O(p^2/\delta^2)$ colors of $G_i$ in ascending order. All the
  nodes on the same level and of the same color are processed in
  parallel. Note by the construction of $\pi'$, when processing a node
  $v\in V$, all out-neighbors $u$ (w.r.t.\ $\pi'$) have already been
  processed.

  Let us now focus on the choice of $x_v$ for a particular node $v\in
  V$. For each $x\in [p]$, we define $\ell_x:=|L_v\cap \calC_{x}|$ to
  be the new list size of $v$ if we set $x_v=x$. We further define
  $b_x$ as the number of neighbors $u$ of $v$ for which $\pi'$ orients
  the edge from $v$ to $u$ and where $x_u=x$. In the end, we will
  define the orientation $\pi$ of all edges of $G$ by using
  orientation $\pi'$ for all edges in $F$ and by arbitrarily orienting
  all edges in $E\setminus F$. If we set $x_v=x$, we can therefore
  upper bound the final number $\beta'(v)$ of out-neighbors $u$ of $v$ (w.r.t.\ orientation
  $\pi$) for which $x_u=x_v$ by $\beta'(v)\leq b_x +
  \delta/p\cdot\beta(v)$. In order to guarantee
  \Cref{eq:listreduction1}, we thus need to find a value $x_v$ for
  which $b_{x_v}+\delta/p\cdot\beta(v)\leq
  (2+\eps)\ell_{x_v}/\alpha_v$, where $\alpha_v$ is defined as in
  \Cref{eq:listreduction1}. Such an $x_v$ exists because
  \[
  \sum_{x=1}^p \left(b_x + \frac{\delta}{p}\cdot \beta(v)\right) =
  \sum_{x=1}^pb_x + \delta\beta(v) \leq (2+\eps)\beta(v) =
  \frac{2+\eps}{\alpha_v}\cdot \sum_{x=1}^p \ell_x.
  \]
  The first inequality follows because $\delta=\eps/2$ and because
  $\sum_{x=1}^p b_i$ is the total number of out-neighbors of $v$
  w.r.t.\ orientation $\pi'$ and it is thus upper bounded by
  $ (2+\delta)\beta(v)$. The last equation follows because
  $\sum_{x=1}^p\ell_x = |L_v|$ and because
  $|L_v|=\alpha_v\beta(v)$. It is clear that for all the nodes that
  are processed in parallel (the nodes of the same color on the same
  level), the color subspace assignment can be done in a single
  communication rounds. The claim of the lemma therefore follows.
\end{proof}

Note that for $\eta=|\calC|$ and sufficiently large lists, the above
theorem directly solves the given list coloring instance. We remark
that for this special case, the time complexity could be improved to
$O(\eta/\eps \cdot \log n)$ by using a more efficient algorithm to
compute a proper coloring of each level of the generalized
$H$-partition. We can now prove \Cref{thm:mainrecursive}. The
following is a restated, slightly more general version of the theorem.

\begin{theorem}
  \label{thm:mainrecursivefull}
  Let $G=(V,E)$ be an $n$-node graph where all edges are oriented and
  for each $v\in V$, let $\beta(v)$ be the out-degree of $v$. Let
  $\eps\in(0,1]$ and $r\geq 1$ be two parameters and assume that each
  node $v$ is assigned a list $L_v$ consisting of colors from some
  color space of size $C$. Then, there is a deterministic distributed
  $O\big(\frac{r\cdot C^{2/r}}{\eps^3}\cdot\log n\big)$-round
  algorithm that computes a proper partial coloring for the given list
  coloring instance such that every node with
  $\beta(v) < (2+\eps)^r\cdot |L_v|$ outputs a color. The algorithm
  uses messages of size $O(\log n + \log C)$.
\end{theorem}
\begin{proof}
  We start with the original list coloring instance and we
  apply \Cref{lemma:orientedreduction} $r$ times to obtain list
  coloring instances with smaller and smaller color spaces. Each time,
  we divide the given color space into several disjoint subspaces. For
  each subspace, we obtain an independent list coloring problem and
  those independent problems can be then be solved in parallel. For
  $i\in\set{0,\dots,r}$, let $C_i$ be the maximum size of the color
  spaces after $i$ list color space reductions (note that
  $C_0=C$). For each node, let $L_{i,v}$ be the color list of $v$ in
  its list coloring instance after $i$ steps and let $\beta_i(v)$ be
  $v$'s out-degree of the corresponding edge orientation after $i$
  steps. We have $L_{0,v} = L_v$ and $\beta_0(v)=\beta(v)$.  We set
  $\eta:=\lceil C^{1/r}\rceil$ and we get the list coloring instances after
  $i$ steps by applying oriented
  $\big(\min\set{\eta,|C_{i-1}|}, 2+\eps\big)$-list color reductions
  to the list coloring instances after $i-1$ steps. By induction on
  $i$, for each $i\geq 1$, we thus have $C_i\leq \max\set{1,C/\eta^i}$ and
  $L_{i,v}/\beta_i(v) \geq L_v/\beta(v)\cdot (2+\eps)^{-i}$. After $r$
  iterations, we therefore have $C_r=1$ and therefore all the lists
  $L_{r,v}$ are also of size $1$. That is, each node $v\in V$ has
  picked exactly one color $c_v$ from its original list $L_v$. If
  $L_v/\beta(v)>(2+\eps)^r$, we have $L_{r,v}/\beta_r(v)>1$ and thus
  $\beta_r(v)=0$. Nodes for which $L_v/\beta(v)>(2+\eps)^r$ therefore
  have no out-neighbor with the same color and if only those nodes
  keep their colors, no two neighbors get the same color. The total
  time complexity follows directly from \Cref{lemma:orientedreduction}
  and the fact that we have $r$ iterations.
\end{proof}

\subsection{Solving Degree + 1 List Coloring}
\label{sec:iterativegeneral}

The coloring result of \Cref{thm:mainrecursive} gives some trade-off
between the list sizes (or number of colors for standard coloring) and
the time complexity. The time complexity depends linearly on $C^{2/r}$
and the list coloring algorithm is thus only efficient if
the value of $r$ is sufficiently large. However, the required list
size to degree ratio grows exponentially with $r$ and for large $r$,
we therefore need large lists. In order to reduce the number of
colors, we use a framework that has been introduced by Barenboim in
\cite{barenboim15} and later refined by Fraigniaud, Kosowski, and
Heinrich in \cite{fraigniaud16}.

The basic idea is as follows. In order to increase the list size to
degree ratio, we first compute some defective (or arbdefective)
coloring to decompose the graph into graphs of smaller maximum
(out-)degree. Let us say we decompose the graph into $K$ parts, which
reduces the maximum degree by a factor $f(K)$ (where $f(K)$ is at most
linear in $K$). Unlike in the previous algorithm, we do not divide the
lists among the different parts, but we use the same lists for
coloring all the parts. If the maximum degree is reduced by a factor
of $f(K)$ and we can use the same color lists, the list size to degree
ratio improves by a factor $f(K)$ and thus if $K$ is chosen
sufficiently large, we can apply \Cref{thm:mainrecursive}.  However,
since now the the different parts are dependent, we have to solve them
consecutively and thus the time complexity will grow linearly with the
number of parts $K$. The main technical result is given by the
following lemma. In order to also be applicable in
\Cref{sec:neighborhoodindep}, the lemma is stated more generally
than what we need in this section.

\begin{lemma}\label{lemma:iterativedegree}
  Let $S$ and $T$ be two functions from the set of undirected graphs
  to the positive reals such that for every graph $G$ and every
  subgraph $H$ of $G$, we have $S(H)\leq S(G)$ and $T(H)\leq T(G)$.
  Assume that we are given a distributed algorithm $\calA$ with the
  following properties. For list coloring instances in graphs
  $G=(V,E)$ with maximum degree $\Delta$, $\calA$ assigns a color $c_v\in L_v$ to each node
  $v\in V$ of degree $\deg(v)\geq \Delta/2$ and with a list $L_v$ of
  size $|L_v|>\deg(v)\cdot S(G)$. Assume further that the running time
  of $\calA$ is at most $T(G)$ if the
  lists contain colors from a color space of size at most
  $\Delta^c$ for an arbitary, but fixed constant $c>0$. Then, if an $O(\Delta^2)$-coloring of $G$ is given,
  any $(\deg(v)+1)$-list coloring instance with colors from a color
  space of size at most $\poly(\Delta)$ can be solved in time
  $O\big((S(G)^2\cdot T(G)+\log^*\Delta)\cdot \log\Delta\big)$. The resulting algorithm is
  deterministic if $\calA$ is deterministic and it works in the
  \CONGEST model if $\calA$ works in the \CONGEST model.
\end{lemma}
\begin{proof}
  By using the same basic idea as in \cite{fraigniaud16}, we show that
  in time $O(S(G)^2\cdot T(G) + \log^*\Delta)$, we can reduce the problem to a
  $(\deg(v)+1)$-list coloring instance on a graph of maximum degree at
  most $\Delta/2$. The lemma then follows by repeating $O(\log\Delta)$
  times.  More formally, we show that given an arbitrary
  $(\deg(v)+1)$-list coloring instance with lists $L_v$ in a graph
  $G=(V,E)$ of maximum degree at most $\Delta$, we can compute a
  partial coloring with the following properties in time
  $O(S(G)^2\cdot T(G) + \log^*\Delta)$:
  \begin{enumerate}[(1)]
  \item Each node $v\in V$ either outputs a color $c_v\in L_v$ or it outputs $c_v=\bot$.
  \item For any two neighbors $u$ and $v$, either $c_u\neq c_ v$ or
    $c_u=c_v=\bot$.
  \item The subgraph of $G$ induced by the nodes $v\in V$ with $c_v=\bot$
    has maximum degree at most $\Delta/2$.
  \end{enumerate}
  Note that because the computed coloring is a proper partial coloring
  (property (2)), the coloring can be extended to a proper coloring of
  all nodes by solving a $(\deg(v)+1)$-list coloring problem on the
  subgraph $H$ induced by the nodes $v$ for which $c_v=\bot$. By
  property (3), this subgraph $H$ has maximum degree at most
  $\Delta/2$. Note  further that by the properties of $S$ an $T$,
  we have $S(H)\leq S(G)$ and $T(H)\leq T(G)$ and by repeating
  $O(\log\Delta)$ times, we thus get a proper coloring of $G$ in time
  $O\big((S(G)^2\cdot T(G)+\log^*\Delta)\cdot\log\Delta\big)$ and the claim of the lemma
  follows.

  In order to construct a partial list coloring algorithm that
  satisfies properties (1)--(3), we first set
  $\eps=1/S(G)$ and we use \Cref{thm:basicdefective} to
  compute an $(\eps/2)$-relative defective $q$-coloring of $G$, where
  $q=O(1/\eps^2)=O(S(G)^2)$. Since we assume that an
  initial $O(\Delta^2)$-coloring of $G$ is given, the defective
  coloring can be deterministically computed in time $O(\log^*\Delta)$
  in the \CONGEST model. For each $x\in [q]$, let $V_x\subseteq V$ be
  the set of nodes with color $x$ in the defective coloring and let $G_x:=G[V_x]$ be the
  subgraph of $G$ induced by the nodes with color $x$. We compute the
  partial coloring $c_v$ for $G$ with properties (1)--(3) as
  follows. 

  We iterate through the $q$ colors of the defective
  coloring. Initially, we set $c_v=\bot$ for all $v\in V$. When
  processing color $x\in [q]$, we consider the graph $G_x$ and we
  assign colors $c_v\neq \bot$ to some nodes $v\in V_x$. For the
  assignment of colors in $G_x$, we consider the following list
  coloring problem. For each node $v\in V_x$, we define the list
  $L_v'$ as the set of colors from $L_v$ that have not yet been
  assigned to a neighbor (in a graph $G_{x'}$ for $x'<x$). We use
  algorithm $\calA$ to compute a partial proper coloring of $G_x$
  w.r.t.\ the lists $L_v'$ for $v\in V_x$. We then set the colors
  $c_v$ accordingly for all $v\in V_x$ to which the partial coloring
  of $G_x$ assigns a color. The choice of lists $L_v'$ guarantees that
  the partial coloring computed for $G$ will be proper and thus
  property (2) is satisfied. Note that if a node $v\in V_x$ has at
  least $\Delta/2$ uncolored neighbors in $G$ before coloring $G_x$,
  the list $L_v'$ is of size $|L_v'|>\Delta/2$ (because we are given a
  $(\deg(v)+1)$-list coloring instance). As the degree of $v$ in $G_x$
  is at most
  \[
  \deg_{G_x}(v) \leq \frac{\eps}{2}\cdot
  \deg(v)=\frac{1}{S(G)}\cdot\frac{\deg(v)}{2} \leq
  \frac{1}{S(G)}\cdot\frac{\Delta}{2} <
  \frac{1}{S(G)}\cdot|L_v'|,
  \]
  algorithm $\calA$ is guaranteed to assign a color to node $v$. This
  proves property (3) of our partial list coloring algorithm for
  $G$. Since the total number of colors $q$ of the defective coloring
  is $O(S(G)^2)$, we only need to invoke $\calA$
  $O(S(G)^2)$ times and thus also the required time complexity
  of $\calA$ follows. The proof is completed by observing that the
  reduction is deterministic and can clearly be implemented with
  messages of size $O(\log n)$.
\end{proof}

We can now prove our main result \Cref{thm:generalDelta}, which shows
that any $(\deg(v)+1)$-list coloring instance can be solved
deterministically in $2^{O(\sqrt{\log\Delta})}\cdot \log n$ rounds in the
\CONGEST model.

\begin{proof}[\bf Proof of \Cref{thm:generalDelta}]
  We set $S(G)=2^{c\sqrt{\log\Delta}}$ and
  $T(G)=2^{c'\sqrt{\log\Delta}}\cdot\log n$ for appropriate constants
  $c,c'>0$.  The claim follows directly from
  \Cref{lemma:iterativedegree} if we can show that there is a
  deterministic $T(G)$-round \CONGEST model algorithm $\calA$ with the
  following properties.  In a graph of maximum degree $\Delta$,
  $\calA$ computes a proper partial list coloring that assigns a color
  to each node $v$ of degree $\Delta/2$ and for which there is a list
  $L_v$ of size $|L_v|\geq\deg(v)\cdot S(G)$. The existence of $\calA$
  follows directly from \Cref{thm:mainrecursivefull} by setting
  $\eps=1$ and $r=\Theta(\sqrt{\log\Delta})$.
\end{proof}

%%% Local Variables:
%%% mode: latex
%%% TeX-master: "main"
%%% End:

%\input{edgecoloring}

\section{Bounded Neighborhood Independence}
\label{sec:neighborhoodindep}

In this section, we give a more efficient algorithm for the
$(\deg(v)+1)$-list coloring problem in graphs of bounded neighborhood
independence. In the following, we assume that we are given a graph
$G=(V,E)$, where the neighborhood independence is bounded by a
constant $\theta\geq 1$. Our algorithm is based on the fact that for
graphs with bounded neighborhood independence, there are more
efficient defective coloring algorithms \cite{barenboim11} and that
defective colorings of such graphs also have stronger properties, as
given by the following lemma.

\begin{lemma}\label{lemma:BNI_defective}
  Let $\theta\geq 1$ and let $G=(V,E)$ be a graph with neighborhood
  independence at most $\theta$. Assume that we are given a
  $d$-defective coloring of the nodes of $G$. Then, for every color
  $x$ of the given defective coloring and every node $v\in V$, the
  number of neighbors $u\in N(v)$ with color $x$ is at most
  $\theta(d+1)$.
\end{lemma}
\begin{proof}
  Let $S_x\subseteq N(v)$ be the nodes of $N(v)$ that have color
  $x$. Because the given coloring has defect $d$, the subgraph
  $G[S_x]$ induced by $S_x$ has degree at most $d$ and can thus be
  colored with $d+1$ colors. The set $S_x$ can thus be partitioned
  into at most $d+1$ independent sets of $G$. Because the neighborhood
  independence of $G$ is at most $\theta$, $N(v)$ can only contain
  independent sets of $G$ of size at most $\theta$ and we thus have
  $|S_x|\leq \theta(d+1)$.
\end{proof}

\begin{corollary}\label{cor:BNI_defective}
  Let $\theta\geq 1$ and $\eps\in(0,1]$ and let $G=(V,E)$ be a graph
  with neighborhood independence at most $\theta$. If an
  $O(\Delta^2)$-coloring of $G$ is given, there is a deterministic
  $O(\log^*\Delta)$-round \CONGEST algorithm to compute a coloring of
  the nodes of $G$ with $O(1/\eps^2)$ colors such that for each $v\in
  V$ and each of the colors $x$, there are at most
  $\max\set{\theta,\eps\deg(v)}$ nodes of color $x$ in $N(v)$.
\end{corollary}
\begin{proof}
  The corollary follows directly from \Cref{lemma:BNI_defective} by applying
  \Cref{thm:basicdefective} with parameter $\lambda =
  \max\set{\frac{1}{2\Delta}, \eps-\frac{1}{\Delta}}
  \geq \frac{\eps}{3}$.
\end{proof}

In \cite{barenboim11}, it was shown that in graphs with bounded
neighborhood independence, it is possible to efficiently compute a
$d$-defective $O(\Delta/d)$-coloring.  This is then used to decompose
a high degree coloring instance into several independent low-degree
instances while only losing a constant factor in the total number of
colors.  In the following, we combine ideas from \cite{barenboim11}
with \Cref{cor:BNI_defective} to obtain such a recursive coloring
algorithm for computing list colorings of graphs with bounded
neighborhood independence. In this case, we need an unoriented variant
of the list color space reduction of \Cref{def:orientedreduction}.  To
make our algorithm work, we have to define a weaker version
of list color space reduction, where only nodes of sufficiently large
degrees are assigned a new color space.

\begin{definition}[Weak List Color Space Reduction]\label{def:unorientedreduction}
  Let $G=(V,E)$ be a graph and assume that we are given a 
  list coloring instance on $G$ with lists $L_v\subseteq \calC$ for some
  color space $\calC$ and all nodes $v\in V$. Assume further that we are given a
  partition $\calC=\calC_1\cup\dots\cup\calC_p$ of the color
  space. 

  Consider an assignment $i_v$ that assigns each node $v\in V$ either
  a color subspace $i_v\in[p]$ or a default value $i_v=\bot$. For each
  $v\in V$ with $i_v\neq \bot$, we define a new list
  $L_v':=L_v\cap \calC_{i_v}$ and we define $\deg'(v)$ to be the
  number of neighboring nodes $u$ for which $i_u=i_v$. 

  The assignment of the values $i_v$ is called an \emph{weak
    $(\eta,\gamma,D)$-list color space reduction} for $\eta> 1$,
  $\gamma\geq 1$, and $D\geq 0$ if $i_v\neq \bot$ for all $v$ for
  which $\deg(v)>D$ and 
  \[
  \forall i\in[p]:\ |\calC_i| \leq \frac{1}{\eta}\cdot |\calC|
  \quad\text{and}\quad
  \forall v\in V:\ i_v\neq \bot\, \rightarrow\,
   \frac{|L_v'|}{\deg'(v)} \geq \frac{1}{\gamma}\cdot \frac{|L_v|}{\deg(v)}.
  \]
\end{definition}

We next describe the algorithm for obtaining such a weak list color
space reduction for graphs with bounded neighborhood independence. Assume that we are given a graph $G=(V,E)$ with
maximum degree $\Delta$. Assume further that we are given an
$O(\Delta^2)$-coloring of $G$ and that each node $v\in V$ is
given a list $L_v$ of colors of some color space $\calC$. The
algorithm has a parameter $\eta\geq 1$ that specifies the factor by
which we reduce the color space. The details of the algorithm are
described in the following.

\paragraph{Weak List Color Space Reduction Algorithm:}
\begin{enumerate}[(1)]
\item Let $\calC=\calC_1\cup\dots\cup\calC_p$ be an arbitrary
  partition of the color space $\calC$ such that $p\leq 2\eta$ and
  $|\calC_i|\leq |\calC|/\eta$ for all $i\in [p]$. Note that such a
  partition is always possible.
\item At the end of the algorithm, each node $v\in V$ outputs an index
  $i_v\in\set{\bot}\cup [p]$ indicating which color subspace $v$
  chooses (if $i_v\neq\bot$). Assume that initially, $i_v$ is set to  $\bot$ for all $v\in V$
\item Use the algorithm of \Cref{cor:BNI_defective} with parameter
  $\eps:=1/p$ to compute a defective coloring of $G$ with $q = O(p^2)$ colors in time $O(\log^*\Delta)$.
\item Iterate through phases $\phi=0,1,2,\dots,O(\log\Delta)$. Define $\delta_0:=\Delta$ and $\delta_{\phi} := \lceil \delta_{\phi-1} / 2\rceil$ for $\phi\geq 1$. Stop after phase $\phi$, where $\delta_{\phi}=1$.
\item In each phase $\phi$, iterate through the $q$ colors $x$ of the
  defective coloring. When processing color $x$, all nodes $v\in V$ of
  color $x$ do the following (in parallel):

  If $i_v=\bot$ (i.e., if $i_v$ is not yet set), $\deg(v)>\theta p$, and if there
  exists an $i^*\in [p]$ such that $v$ has at most $\delta_{\phi}$
  neighbors $u$ with $i_u=i^*$ and
  $\delta_{\phi}+\frac{\deg(v)}{p}\leq
  3\cdot|L_v\cap\calC_{i^*}|\cdot\frac{\deg(v)}{|L_v|}$,
  set $i_v := i^*$.
\end{enumerate}

\begin{lemma}\label{lemma:BNIspacereduction}
  Let $G=(V,E)$ be a graph with maximum degree at most $\Delta$ and
  neighborhood independence at most $\theta\geq 1$. 
  Assume that we are given a list coloring instance with lists
  $L_v\subseteq \calC$ for each  $v\in V$. For a parameter
  $\eta>1$, the above algorithm computes a weak
  $(\eta,3\theta,2\theta\eta)$-list color space reduction in
  $O(\eta^2\log\Delta)$ rounds deterministically in the \CONGEST
  model.
\end{lemma}
\begin{proof}
  We first show that the algorithm assigns an color subspace $i_v\in[p]$ to
  each node $v\in V$ with $\deg(v)>\theta p$. For the sake of contradiction, assume that there
  is a node $v\in V$ with $\deg(v)>\theta p$ to which the algorithm
  does not assign a color subspace
  $i_v\in[p]$. Node $v$ in considered in each phase $\phi$ and it
  chooses a color subspace in the first phase $\phi$ in which there is
  an $i^*\in [p]$ such that $v$ has at most $\delta_\phi$ neighbors
  $u$ with $i_u=i^*$ and for which
  \begin{equation}\label{eq:BNIassignmentcondition}
    \delta_{\phi} + \frac{\deg(v)}{p} \leq 3\cdot\frac{\deg(v)}{|L_v|}\cdot |L_{v}\cap\calC_{i^*}|.    
  \end{equation}
  If $v$ does not pick a color space, this condition is not satisfied
  for any phase $\phi$. To prove a contradiction, we consider the
  assignment to the neighbors $u\in N(v)$ of $v$ at the end of the
  algorithm. For each $i\in [p]$, let $d_i$ be the number of node
  $u\in N(v)$ for which $i_u=i$ at the end of the algorithm. Note that
  clearly $\sum_{i=1}^p d_i \leq |N(v)|=\deg(v)$. Further, for each
  $i\in[p]$, we define $\ell_i:=|L_v\cap\calC_i|$ to be the new list
  size of $v$ if $v$ sets $i_v=i$. We have
  \[
  \sum_{i=1}^p \left(2d_i+\frac{\deg(v)}{p}\right) \leq
  3\deg(v) = 
  3\cdot\frac{\deg(v)}{|L_v|}\cdot |L_v| = 
  3\cdot\frac{\deg(v)}{|L_v|}\cdot \sum_{i=1}^p \ell_i.
  \]
  Therefore, there exists an $i^*$ for which
  $2d_{i^*}+\frac{\deg(v)}{p} \leq
  3\frac{\deg(v)}{|L_v|}\cdot\ell_{i^*}$.
  Let $\phi^*$ be the last phase for which
  $d_{i^*}\leq \delta_{\phi^*}$. We have
  $\delta_{\phi^*}\leq 2 d_{i^*}$ and therefore
  $\delta_{\phi^*}+\frac{\deg(v)}{p}\leq3\cdot\frac{\deg(v)}{|L_v|}\cdot
  \ell_{i^*}$.
  Hence, if $i_v$ is still set to $\bot$ in phase $\phi^*$ when
  processing node $v$, the algorithm can set $i_v=i^*$ in phase $\phi^*$ (note that
  the number of neighbors $u\in N(v)$ with $i_u=i^*$ in phase $\phi^*$
  is at most $d_i\leq \delta_{\phi^*}$). This is a contradiction to
  the assumption that $v$ does not set its $i_v$ to a value in $[p]$
  and we have thus shown that at the end of the algorithm, all nodes
  $v$ with $\deg(v)>\theta p$ have $i_v\in[p]$.

  For the remainder of the proof, consider some node $v$ for which the
  algorithm sets $i_v=i^*$. Note that we have $\deg(v)>\theta p$ as the
  algorithm would otherwise not set the value $i_v\neq \bot$. Assume
  that node $v$ sets its value $i_v$ to $i^*$ in phase $\phi^*$. Let
  $S(v)$ be the set of neighbors $u\in N(v)$ of $v$ for which at the
  end of the algorithm $i_u=i_v=i^*$. We next bound the size of $S(v)$
  as 
  \begin{equation}
    \label{eq:degprime}
    |S(v)|\leq \theta\cdot\left(\delta_{\phi^*} + \frac{\deg(v)}{p}\right).    
  \end{equation}
  In order to prove \Cref{eq:degprime}, we consider the subgraph
  $G[S(v)]$ of $G$ induced by the nodes in $S(v)$ and we show that
  $G[S(v)]$ can be colored with a small number of colors. Together
  with the neighborhood independence of $G$, this implies that $S(v)$
  is not too large. To bound the chromatic number of $G[S(v)]$, we
  compute a low out-degree edge orientation of $G[S(v)]$. Consider two
  nodes $u,w\in S(v)$ such that $\set{u,w}\in E$. If $u$ and $w$ set
  $i_u=i_w=i^*$ at the same time, we orient the edge $\set{u,w}$ from
  the node with the lower ID to the node with the larger
  ID. Otherwise, we orient the edge from the node that sets its color
  subspace to $\calC_{i^*}$ later to the node that set its color
  subspace to $\calC_{i^*}$ first. Note that the resulting orientation
  is clearly acyclic. 

  To bound the outdegree of a node $u\in S(v)$, let us first bound the
  number of neighbors $w$ of $u$ in $G[S(v)]$ that set their color space
  to $\calC_{i^*}$ before $u$. Before node $v$ sets $i_v=i^*$, the
  number of nodes in $S(v)$ is at most $d_{\phi^*}$ and therefore for
  all nodes $u$ in $S(v)$ that choose the color subspace before $v$,
  the number of nodes $w\in S(v)$ that choose the color subspace
  before $u$ is bounded by $d_{\phi^*}$. Every node $u\in S(v)$ that
  sets $i_u=i^*$ together with $v$ or after $v$ does this in a phase
  $\phi\geq \phi^*$. Hence, for such nodes $u$, when $u$ sets $i_u$ to
  $i^*$, it has at most $d_{\phi}\leq d_{\phi^*}$ neighbors $w$ with
  $i_w=i^*$. For every node $u\in S(v)$, the number of neighbors $w\in
  G[S(v)]$ that set $i_w=i^*$ before $u$ sets $i_u=i^*$ is therefore
  upper bounded by $d_{\phi^*}$.

  We next bound the number of neighbors $w$ in $G[S(v)]$ of a node
  $u\in S(v)$ that choose their color subspace at the same time as
  $u$. Note that only nodes that have the same color in the defective
  coloring computed in step (3) of the algorithm can choose their
  color subspace together. Note that by \Cref{cor:BNI_defective} (and
  the choice of $\eps=1/p$), for each of the colors $x$ of the
  defective coloring, the number of nodes $u\in N(v)$ with color $x$
  is at most $\max\set{\theta,\deg(v)/p}\leq \deg(v)/p$ (because
  $\deg(v)>\theta p$). For each $u$ in $S(v)$, the number of nodes
  $w\in S(v)\setminus\set{u}$ that choose their color subspace
  together with $u$ is therefore upper bounded by
  $\deg(v)/p-1$. Consequently, the computed acyclic edge orientation
  of $G[S(v)]$ has out-degree at most $d_{\phi^*} + \deg(v)/p -1$.
  This implies that $G[S(v)]$ has chromatic number at most
  $d_{\phi^*} + \deg(v)/p$ and thus that the nodes in $S(v)$ can be
  partitioned into $k\leq d_{\phi^*} + \deg(v)/p$ independent sets (of
  $G$). Because the neighborhood independence of $G$ is bounded by
  $\theta$, each of these independent sets is of size at most $\theta$
  and \Cref{eq:degprime} therefore follows. It now follows directly by
  combining \Cref{eq:degprime} with \Cref{eq:BNIassignmentcondition}
  that the algorithm computes a weak $(\eta,3\theta,2\theta\eta)$-list
  color space reduction as claimed. The algorithm is deterministic and
  it can clearly be implemented in the \CONGEST model. The time
  complexity of the algorithm is dominated by iterating through the
  $O(\log\Delta)$ phases and the $O(p^2)=O(\eta^2)$ colors of the
  defective coloring in each phase.
\end{proof}

\begin{lemma}\label{lemma:recursiveBNI}
  Let $G=(V,E)$ be a graph with maximum degree $\Delta$ and
  neighborhood independence at most $\theta\geq 1$. Assume 
  that each node $v\in V$ is assigned a list $L_v\subseteq \calC$ of
  colors from some color space $\calC$ of size $|\calC|=C$. Assume
  further that an $O(\Delta^2)$-coloring of $G$ is given. Then, there is a
  deterministic distributed
  $O\big(\theta\cdot r\cdot C^{2/r}\log\Delta\big)$-round algorithm that computes
  a proper partial list coloring such that each node $v$
  with a list of size $|L_v| > (3\theta)^{r-1}\cdot \deg(v)$ is assigned a color.
\end{lemma}
\begin{proof}
  We first show that for any graph $G=(V,E)$ and an arbitrary
  parameter $d$, it is possible to compute a partial proper list
  coloring of the nodes $v\in V$ with $\deg(v)\leq d$ in time
  $O(d+\log^*\Delta)$. The partial list coloring assigns a color
  $v\in L_v$ to all such nodes with $|L_v|>\deg(v)$. Let $V_d$ be the
  set of nodes with $\deg(v)\leq d$ and let $G[V_d]$ be the subgraph
  of $G$ induced by the nodes in $V_d$. In the graph $G[V_d]$ every
  node has clearly degree at most $d$ and by (for example) using the
  algorithm of \cite{barenboim09,spaa09} and the initial
  $O(\Delta^2)$-coloring of $G$ (and thus also of $G[V_d]$), we can
  therefore compute a $(d+1)$-coloring of $G[V_d]$ in time
  $O(d + \log^*\Delta)$.  We can then iterate through the $d+1$ colors
  of this coloring and assign list colors greedily. When doing this,
  every node $v$ with $|L_v|>\deg(v)$ definitely succeeds in getting a
  color.

  We can now prove the lemma by induction on $r$. Let us first
  consider the case $r\leq 2$ as the base of the induction. In this case,
  we can spend time $O(C + \log^*\Delta)$ and since any node $v$ with
  $\deg(v)>|L_v|$ clearly has $\deg(v)< C$, for $r\leq 2$, the claim
  of the lemma directly follow from the above
  $O(d+\log^*\Delta)$-round algorithm.

  To prove the induction step, let us therefore consider some $r>2$
  and assume that the claim of the lemma holds for smaller $r$. We
  define $\eta := C^{1/r}$ and we partition the set of nodes $V$ into
  low-degree nodes $V_L$ and high-degree nodes $V_H$. We let $V_L$ be
  the set of node $v\in V$ for which $\deg(v)\leq 2\theta\eta$ and we
  accordingly define $V_H := V\setminus V_L$. We apply the above weak
  list color space reduction algorithm to $G$ with parameter
  $\eta$. The algorithm partitions the color space $\calC$ into
  subspaces $\calC_1,\dots,\calC_p$ of size at most $C^{(r-1)/r}$.
  Because the nodes in $V_H$ have degree more than $2\theta\eta$, by
  \Cref{lemma:BNIspacereduction}, all nodes $v\in V_H$ are assigned a
  subspace $\calC_{i_v}$ (for $i_v\in [p]$). Let $L_v'=L_v\cap \calC_{i_v}$ be the new
  list of a node $v\in V_H$ and let $\deg'(v)$ be the degree of $v$
  in the subgraph of $G$ induced by all the nodes in $V_H$ that were
  assigned to the same color subspace $\calC_{i_v}$. By
  \Cref{lemma:BNIspacereduction}, we have
  $|L'_v|/\deg'(v)\geq |L_v|/(3\theta\deg(v))$. Hence, all nodes $v\in V_H$
  for which $|L_v|>(3\theta)^{r-1}\cdot \deg(v)$, we have
  $|L_v'|>(3\theta)^{r-2}\cdot\deg'(v)$. For each of the color subspaces
  $\calC_i$, we can therefore apply the induction hypothesis to list
  color the subgraphs induced by the nodes $v\in V_H$ that chose
  $i_v=i$. The list coloring problem on each of these subgraphs has a
  color space of size at most $C^{(r-1)/r}$. Hence by applying
  the induction hypothesis for $r-1$, we can list 
  color each of those subgraphs in time
  $O((r-1)C^{2/r}\log\Delta)$. All nodes $v\in V_H$ with
  $|L_v'|>(3\theta)^{r-2}\deg'(v)$ and thus in particular all nodes $v\in V_H$
  with $|L_v|>(3\theta)^{r-1}\deg(v)$ are assigned a color.

  It therefore just remains to color the nodes $v\in V_S$ for which
  $|L_v|>(3\theta)^{r-1}\deg(v)$ and thus in particular
  $|L_v|>\deg(v)$. Note that by definition of $V_S$, the graph induced
  by the nodes in $V_S$ has degree at most $2\theta\eta$. Further, for
  each node $v\in V_S$ with $|L_v|>\deg(v)$, also for the remaining
  list coloring problem on $G[V_S]$, the list size of $v$ exceeds its
  degree (for each already colored neighbor in $V_H$, the degree of
  $v$ drops by $1$ and the list size of $v$ drops by at most
  $1$). Hence, the remaining problem on $G[V_S]$ can be solved in time
  $O(\theta\eta + \log^*\Delta)=O(\theta C^{1/r} + \log^*\Delta)$ by
  using the $O(d+\log^*\Delta)$ algorithm described in the first
  paragraph of this proof. Thus the claim of the lemma follows.
\end{proof}

Based on \Cref{lemma:recursiveBNI}, we can now prove \Cref{thm:BNI} the main theorem
of this section. The theorem is almost an immediate
consequence of \Cref{lemma:recursiveBNI,lemma:iterativedegree}.

\begin{proof}[\bf Proof of \Cref{thm:BNI}]
  We start by computing an $O(\Delta^2)$-coloring of $G$. By using the
  algorithm of \cite{linial92}, this can be done in time
  $O(\log^* n)$. Based on this coloring, we can now apply
  \Cref{lemma:recursiveBNI}. When setting
  $r=\Theta(\sqrt{\log\theta\cdot\log \Delta})$ and using that
  $C\leq\poly(\Delta)$, the lemma implies that there is an
  $O(r\theta
  C^{2/r}\log\Delta)=2^{O(\sqrt{\log\theta\cdot\log\Delta})}$-time
  algorithm that computes partial proper list coloring such that all
  nodes $v\in V$ with
  $|L_v|>(3\theta)^{r-1}\cdot
  \deg(v)=2^{\Theta(\sqrt{\log\theta\cdot\log\Delta})}\cdot \deg(v)$
  are assigned a color. The algorithm thus satisfies the requirements
  of \Cref{lemma:iterativedegree} for
  $S(G)=T(G)=2^{O(\sqrt{\log\theta\cdot\log\Delta})}$ and for graphs
  with neighborhood independence at most $\theta$, we get a
  $2^{O(\sqrt{\log\theta\cdot\log\Delta})}$-time deterministic
  distributed algorithm to solve an arbitrary list coloring instance
  with lists of size $|L_v|>\deg(v)$ for all $v\in V$ and colors from
  a range of size $\poly(\Delta)$.
\end{proof}

%%% Local Variables:
%%% mode: latex
%%% TeX-master: "main"
%%% End:

\section{Acknowledgements}

I am grateful to various colleagues for several discussions on 
distributed coloring algorithms that helped shaping my understanding of
the problem. I am particularly thankful to Mohsen Ghaffari, Yannic
Maus, and Jara Uitto for discussions on the relevance of
finding better list coloring algorithms even for relatively large list
sizes (such as the ones considered in \Cref{sec:recursive}).

\bibliographystyle{alphaabbr}
\bibliography{references}

\end{document}